\documentclass[12pt]{amsart}

\newtheorem{thm}{Theorem}[section]

\newtheorem{lem}[thm]{Lemma}
\newtheorem{prop}[thm]{Proposition}
\newtheorem*{prob*}{Problem}

\newtheorem*{thm*}{Theorem}

\theoremstyle{definition}

\newtheorem*{defn*}{Definition}
\newtheorem{rem}[thm]{Remark}
\numberwithin{equation}{section}

\newcommand{\C}{\mathbb C}

\newcommand{\Y}{\mathbb Y}
\newcommand{\Z}{\mathbb Z}
\newcommand{\Zp}{\mathbb Z_{\geq 0}}

\newcommand{\DC}{\mathcal{D}}
\newcommand{\MC}{\mathcal{M}}

\newcommand{\ZC}{\mathcal{Z}}

\DeclareMathOperator{\Schur}{Schur}
\DeclareMathOperator{\Plancherel}{Plancherel}

\DeclareMathOperator{\Pf}{Pf}

\begin{document}
\title[Matrix kernels for partitions]
 {\bf{Matrix kernels for measures on partitions}}

\author{Eugene Strahov}

\date{April 7, 2008}

\thanks{Department of Mathematics, The Hebrew University of
Jerusalem, Givat Ram, Jerusalem 91904. E-mail:
strahov@math.huji.ac.il. Supported by US-Israel Binational Science Foundation (BSF) Grant No. 2006333.\\
}


\commby{}
\begin{abstract}
We consider the problem of computation of the correlation
functions for the $z$-measures  with the deformation (Jack)
parameters $2$ or $1/2$. Such measures on partitions are
originated from the representation theory of the infinite
symmetric group, and in many ways are similar to the ensembles of
Random Matrix Theory of $\beta=4$ or $\beta=1$ symmetry types.
For a certain class of such measures we show that correlation
functions can be represented as Pfaffians including $2\times 2$
matrix valued kernels, and compute these kernels explicitly. We
also give contour integral representations for correlation kernels
of closely connected measures on partitions.
\end{abstract}
\maketitle
\section{Introduction}
\subsection{Preliminaries and formulation of the
problem}\label{SectionFormulation} It was observed by many authors
that number of problems of statistical mechanics and of
combinatorial probability can be reduced to investigation of random
Young diagrams. In particular, increasing subsequences problems,
certain last-passage percolation and growth models, random  tilings
are all examples of such a situation, see Baik, Deift, and Johansson
\cite{BaikDeiftJohansson}, Baik and Rains
\cite{BaikRains1,BaikRains2, BaikRains3}, Johansson
\cite{Johansson1, Johansson2, Johansson3}, Pr\"{a}hover and Spohn
\cite{PrahoverSpohn}, Ferrari and Pr\"{a}hover
\cite{FerrariPrahover}, and references therein.

For instance, in  the analysis of the polynuclear growth (PNG) model
in one spatial dimension the relation with measures on Young
diagrams arises as follows, see Ferrari and Pr\"{a}hover for a
recent review \cite{FerrariPrahover}.  Following Ferrari and
Pr\"{a}hover, let us describe the growing surface at time $t$ by an
integer-valued height function $x\rightarrow h(x,t)\in\Z$. For fixed
time $t\in\mathbb{R}$, one considers the height profile
$x\rightarrow h(x,t)$. The height $h$, as $x$ increases, has jumps
of height one at the discontinuity points. By definition, the
dynamics of the model has a deterministic and a stochastic part. The
deterministic part is described by the condition that the up-steps
move to the left with unit speed, and the down steps move to the
right with unit speed. When a pair of up-and down steps collide,
they disappear. The stochastic parts describe  nucleation events
(creation of up-down steps on the surface). It is assumed that once
a pair of up-down steps is created, it follows the deterministic
dynamics. The main problem in this context is to describe the
fluctuations of height profile $h(x,t)$ as time grows. Typically one
considers two cases of particular interest, called the PNG droplet
and the flat PNG. In both cases the investigation of  $h(x,t)$  at
point $x=0$ is reduced to that of the first row of a random Young
diagram. In the first situation (the PNG droplet) the relevant
measure on Young diagrams is the Poissonized Plancherel measure, and
in the second situation (the flat PNG) the relevant measure on Young
diagrams can be understood as the Poissonized Plancherel measure
with Jack parameter $1/2$. Note also that the Plancherel measures
with Jack parameters $2$, $1/2$ also appear in the analysis of the
symmetrized increasing subsequence problems, see Baik and Rains
\cite{BaikRains2}, Forrester, Nagao and Rains \cite{forrester}.

The Plancherel measures mentioned above can be understood  as limits
of the $z$-measures  $\MC_{z,z',\theta,\xi}$ which are probability
measures on the set $\Y$ of all Young diagrams. These measures arise
in the context of the representation theory of the infinite
symmetric group. Namely, let $S(\infty)$ be the infinite symmetric
group whose elements are finite permutations of $\{1,2,\ldots\}$.
Kerov, Olshanski, and Vershik \cite{kerov1} (see also a review paper
by Olshanski \cite{olshanski2003}) constructed a family
$\{T_z:z\in\C\}$ of unitary representations of two copies $G$ of
$S(\infty)$. Such representations were called the generalized
regular representations of the infinite symmetric group. Each
representation $T_{z}$ with $z\neq 0$ can be realized in the Hilbert
space $L^2\left(\mathfrak{S},\mu_t\right)$, where $\mathfrak{S}$ is
a compatification of $S(\infty)$ (called the space of virtual
permutations), $\{\mu_t\}$ is a family of $G$-quasiinvariant
measures on virtual permutations, and $t=|z|^2$. Denote by
$\chi^{z}$ the character of $T_z$. According to Kerov, Olshanski and
Vershik \cite{kerov1}, $\chi^z$ is a central positive definite
function on $S(\infty)$ such that, for any $n=1,2,\ldots$, its
restriction to $S(n)$ is
$$
\chi_n^{z}= \sum\limits_{\lambda\in\Y_n}\MC_{z,\bar{z}}^{(n)}
(\lambda)\frac{\chi^{\lambda}}{\chi^{\lambda}(e)}.
$$
Here $\chi^{\lambda}$ denotes the irreducible character of $S(n)$
parameterized by the Young diagram $\lambda$, and $\Y_n$ is the
set of all Young diagrams with $n$ boxes. It can be shown that for
every $n=1,2,\ldots$ $\MC_{z,\bar{z}}^{(n)}$ is a probability
distribution on the set $\Y_n$.

Borodin and Olshanski \cite{borodin1,borodin2,borodin3} considered
a $3$-parameter family of measures $\MC_{z,z',\xi}$ (called
$z$-measures) on the set of all Young diagrams $\Y$. This family
was constructed by mixing of $\MC_{z,z'}^{(n)}$, see Borodin and
Olshanski \cite{borodin1} for details. Under certain conditions on
parameters, $\MC_{z,z',\xi}$ is a probability measure on $\Y$.
 Borodin and
Olshanski \cite{borodin1} showed that $\MC_{z,z',\xi}$  defines a
determinantal point process, and that $\MC_{z,z',\xi}$ can be
studied on the same level  as the Gaussian Unitary Ensemble of
Random Matrix Theory (see, for example, Mehta \cite{mehta}, Deift
\cite{Deift1}, Forrester \cite{forrester0} as basic references on
ensembles of random matrices; for a review of determinantal point
processes see Soshnikov \cite{soshikov}).

As it follows from   Kerov \cite{kerov}, Borodin and Olshanski
\cite{BO} it is natural to consider a deformation
$\MC_{z,z',\theta,\xi}$ of $\MC_{z,z',\xi}$, where $\theta>0$ is
called the parameter of deformation (or the Jack parameter). As it
shown in Borodin and Olshanski \cite{BO}, $\MC_{z,z',\theta,\xi}$
are in many ways similar to log-gas (random-matrix) models with
arbitrary $\beta=2\theta$. In particular, if $\theta=2$ or
$\theta=1/2$ one expects that $\MC_{z,z',\theta,\xi}$ will lead to
Pfaffian point processes, similar to ensembles of Random Matrix
Theory of $\beta=4$ or $\beta=1$ symmetry types. Thus, it is a
natural problem to compute the correlation functions for
$z$-measures with the deformation parameters $\theta=2$ or $1/2$.

More explicitly, set
$$
\DC_2(\lambda)=\{\lambda_i-2i\}.
$$
That is, $\DC_2(\lambda)$ is an infinite subset of $\Z$
corresponding to the Young diagram $\lambda$, and containing
infinitely many integers. Let $X=(x_1,\ldots,x_n)$ be a subset of
$\Z$ consisting of $n$ distinct points, and define
$$
\varrho_{z,z',2,\xi}(X)=\MC_{z,z',2,\xi}\left(\{\lambda|X\subset\DC_2(\lambda)\}\right).
$$
The problem is to give an explicit formula for
$\varrho_{z,z',2,\xi}(X)$, i.e. to compute the correlation
function of $\MC_{z,z',2,\xi}$. Note that as soon as
$\MC_{z,z',\theta,\xi}$ and $\MC_{z,z',1/\theta,\xi}$ are related
by
$$
\MC_{z,z',\theta,\xi}(\lambda)=\MC_{-z/\theta,-z'/\theta,1/\theta,\xi}(\lambda'),
$$
it is enough to consider the $z$-measures with the deformation
parameter $\theta=2$. Defining $\varrho_{z,z',2,\xi}$ as above we
do not require positivity of $\MC_{z,z',2,\xi}$.
\subsection{Summary of results}
\subsubsection{A formula for the correlation function of
$\MC_{z,z-1,2,\xi}$} Theorem \ref{Theoremzmeasure2} gives an
explicit formula for the correlation function of
$\MC_{z,z-1,2,\xi}$. This result shows that the correlation
function of $\MC_{z,z-1,2,\xi}$ can be written as a Pfaffian of a
$2\times 2$ matrix valued kernel. The entries of this kernel are
expressible in terms of one  function, $S_{z,z-1,2,\xi}(x,y)$,
which admits a contour integral representation, see Theorem
\ref{Theoremzmeasure2}, a). Theorem 3.1, b) gives a formula for
$S_{z,z-1,2,\xi}(x,y)$ in terms of the Gauss hypergeometric
functions.

Thus, Theorem \ref{Theoremzmeasure2} provides a partial solution
of the problem stated in Section \ref{SectionFormulation}, and it
is the main result of the present paper.
\subsubsection{A contour
integral representation for correlation functions of Pfaffian
Schur measures} In fact, results of Theorem \ref{Theoremzmeasure2}
can be obtained from a more general algebraic fact. We consider
measures on partitions which can be understood as Pfaffian
analogues of Okounkov's Schur measures, and derive  contour
integral representations for correlation kernels in Theorem
\ref{MainTheorem}. Then Theorem \ref{Theoremzmeasure2} is obtained
by considering a specific specialization of the algebra of
symmetric functions. As it will be evident from the proof of
Theorem \ref{Theoremzmeasure2}, the correlation function of
$\MC_{z,z',2,\xi}$ in the case of arbitrary parameters $z,z'$
cannot be deduced from Theorem \ref{MainTheorem}, and a different
approach is needed.
\subsubsection{Correlation functions for the Plancherel measure
with the deformation parameter $\theta=2$} When both $z,z'$ go to
infinity, the $z$-measures with general parameter $\theta>0$ have
limits. These limits can be understood  as Plancherel measures
with general parameter $\theta>0$. Theorem
\ref{TheoremPlancherelMeasure2}, a) gives the contour integral
representation for the correlation kernel of the Plancherel
measure with the deformation parameter $\theta=2$, and Theorem
\ref{TheoremPlancherelMeasure2}, b) gives a formula for these
correlation kernels in terms of Bessel functions.
\subsection{Remarks on related works}
\subsubsection{}It is known that a number of important measures on partitions
lead to determinantal point processes. The most well-known example
of measures of this type is the poissonization of the Plancherel
measure, which was considered by many researches, see Baik, Deift
and Johansson \cite{BaikDeiftJohansson}, Borodin, Okounkov and
Olshanski \cite{BorodinOkounkovOlshanski}, the review articles by
Deift \cite{Deift}, and by van Moerbeke \cite{vanMoerbeke}. The
$z$-measures (with $\theta=1$)  is another family of measures
leading to a determinantal point process. The fact that
$z$-measures define a determinantal process was first proved by
Borodin and Olshanski \cite{borodin1}. The paper by Borodin and
Olshanski \cite{borodin1} contains an explicit computation of the
correlation kernel, and shows that the correlation kernel can be
expressed through Gauss hypergeometric functions. The $z$-measures
on partitions were also studied in Borodin and Olshanski
\cite{borodin2,borodin3}, Okounkov \cite{okounkov5}, Borodin,
Olshanski, and Strahov \cite{borodin4}. Note that for special
values of parameters $z$-measures turn into discrete orthogonal
polynomial ensembles. Such ensembles are related with different
exactly solvable probabilistic models, see Borodin and Olshanski
\cite{borodin2}, and the references therein.
\subsubsection{}
Okounkov \cite{okounkov1} introduced a family of measures on
partitions called the Schur measures. This family of measures
includes the $z$-measures (with $\theta=1$) and the Poissonized
Plancherel measure as special cases. Okounkov
\cite{okounkov2,okounkov3} discusses different applications of
such measures. For relations of the Schur measures  with the space
of free fermions see Okounkov \cite{okounkov1}, a recent review
article by Harnard and Orlov \cite{harnard}, and the references
therein.
\subsubsection{}
The Schur processes  introduced in Okounkov and Reshetikhin
\cite{okounkov4}, and the periodic Schur measures introduced in
Borodin \cite{BorodinPeriodic} are generalizations of the Schur
measures, which also define determinantal point processes. The
correlation kernels of the Schur measures, of the Schur processes
and of the periodic Schur processes have contour integral
representations. As it was demonstrated in Okounkov
\cite{okounkov2,okounkov3}, Okounkov and Reshetikhin
\cite{okounkov4}, Borodin \cite{BorodinPeriodic} the contour
integral representations for the correlation kernels have many
advantages, in particular, such representations are very
convenient for the asymptotic analysis.
\subsubsection{}
The class of measures considered in Section
\ref{SectionPfaffianSchur} of the present paper was first
introduced by Rains \cite{rains} in the context of symmetrized
increasing subsequence problems, and in certain sense these
measures  are Pfaffian analogues of Okounkov's Schur measures.
Paper by Rains \cite{rains} also contains the idea of computation
of the correlation kernels.  Borodin and Rains
\cite{borodinrains}, Vuleti$\acute{\mbox{c}}$ \cite{vuletic} have
considered different classes of measures (the Pfaffian Schur
processes in Borodin and Rains \cite{borodinrains}, and the
shifted Schur processes in Vuleti$\acute{c}$ \cite{vuletic}) whose
correlation functions are expressible as Pfaffians of matrix
valued kernels. It can be verified that the results of Theorem
\ref{MainTheorem} can be deduced from that of Theorem 3.3 in
Borodin and Rains \cite{borodinrains}. However, the author of this
paper was not able to find the explicit formula for
$S_{\Schur,2}(x,y)$ (see Theorem \ref{MainTheorem}) in the
literature. In the
present paper we give an independent proof of Theorem \ref{MainTheorem}.\\
\subsubsection{}
The correlation functions for measures, which are in many ways
similar to the Plancherel measures with the deformation parameters
$\theta=1/2, 2$, were previously computed by different authors, see
Ferrari \cite{ferrari}, Lemma 5.2, Sasamoto and Imamura
\cite{sasamoto}, and Forrester, Nagao, and Rains \cite{forrester},
Section 3. In particular, Ferarri \cite{ferrari} computes a
correlation function for the following measure on even-rows Young
diagrams
$$
\MC(\lambda)=e^{-\eta}\frac{\eta^n}{n!}\frac{\dim\lambda}{Z_n},
$$
where $Z_n$ is the number of standard Young tableaux with $2n$
entries and even rows. The measure just written above can be
understood as the Plancherel measure with the Jack parameter $1/2$.
The result of Ferrari can be obtained from that of Section 7 of the
present paper as follows. Theorem 7.1 gives a correlation kernel for
the Plansherel measure defined on  the point configurations $
\DC_2(\lambda)=\{\lambda_i-2i\} $. One can also define a correlation
function with point configurations $
\tilde{\DC}_2(\lambda)=\{\lambda_i-2i+1,\lambda_i-2i\}$, and this
correlation function can be determined from the formula in Theorem
7.1. By particle-hole involution this gives a correlation function
for point configurations of the form $\{-\lambda_i'+i-1\}_{i\geq
1}^{\infty}$.  Then the relation
$M^{(n)}_{\theta}(\lambda)=M^{(n)}_{\frac{1}{\theta}}(\lambda')$
enables to obtain the correlation function for the Plancherel
measure with the Jack parameter $\theta=1/2$.

\subsubsection{}Special cases of the $z$-measures with $\theta=2,1/2$
can be understood as discrete symplectic and orthogonal ensembles.
For the explanation of this relation, and for the results on
correlation functions of the corresponding discrete ensembles see
Borodin and Strahov
\cite{borodinstrahov1}.\\
\textbf{Acknowledgment.} I am  grateful to Alexei Borodin and to
Patrik Ferrari for
 their interest in this work, and for helpful discussions. I also
 very grateful to reviewers for many valuable comments.

\section{The $z$-measures with the general parameter $\theta>0$}
We use Macdonald \cite{macdonald} as a basic reference for the
notations related to integer partitions and to symmetric
functions. In particular, every decomposition
$$
\lambda=(\lambda_1,\lambda_2,\ldots,\lambda_l):\;
n=\lambda_1+\lambda_2+\ldots+\lambda_{l},
$$
where $\lambda_1\geq\lambda_2\geq\ldots\geq\lambda_l$ are positive
integers, is called an integer partition. We identify integer
partitions with the corresponding Young diagrams, and denote the
set of all Young diagrams by $\Y$.  The set of Young diagrams with
$n$ boxes  is denoted by $\Y_n$. Thus
$$
\Y=\bigcup\limits_{n=0}^{\infty}\Y_n.
$$
Following Borodin and Olshanski \cite{BO}, Section 1, let
$\MC_{z,z',\theta}^{(n)}$ be a complex measure on $\Y_n$ defined
by
\begin{equation}\label{EquationVer4zmeasuren}
\MC_{z,z',\theta}^{(n)}=\frac{n!(z)_{\lambda,\theta}(z')_{\lambda,\theta}}{(t)_nH(\lambda,\theta)H'(\lambda,\theta)},
\end{equation}
where $n=1,2,\ldots $, and where we use the following notation
\begin{itemize}
    \item $z,z'\in\C$ and $\theta>0$ are parameters, the parameter
    $t$ is defined by
    $$
    t=\frac{zz'}{\theta}.
    $$
    \item $(t)_n$ stands for the Pochhammer symbol,
    $$
    (t)_n=t(t+1)\ldots (t+n-1)=\frac{\Gamma(t+n)}{\Gamma(t)}.
    $$
    \item
    $(z)_{\lambda,\theta}$ is a multidimensional analogue of the
    Pochhammer symbol defined by
    $$
    (z)_{\lambda,\theta}=\prod\limits_{(i,j)\in\lambda}(z+(j-1)-(i-1)\theta)
    =\prod\limits_{i=1}^{l(\lambda)}(z-(i-1)\theta)_{\lambda_i}.
    $$
     Here $(i,j)\in\lambda$ stands for the box in the $i$th row
     and the $j$th column of the Young diagram $\lambda$, and we
     denote by $l(\lambda)$ the number of nonempty rows in the
     Young diagram $\lambda$.
    \item
    $$
    H(\lambda,\theta)=\prod\limits_{(i,j)\in\lambda}\left((\lambda_i-j)+(\lambda_j'-i)\theta+1\right),
   $$
   $$
     H'(\lambda,\theta)=\prod\limits_{(i,j)\in\lambda}\left((\lambda_i-j)+(\lambda_j'-i)\theta+\theta\right),
   $$
      where $\lambda'$ denotes the transposed diagram.
\end{itemize}
\begin{prop}\label{PropositionHH}
The following symmetry relations hold true
$$
H(\lambda,\theta)=\theta^{|\lambda|}H'(\lambda',\frac{1}{\theta}),\;\;(z)_{\lambda,\theta}
=(-\theta)^{|\lambda|}\left(-\frac{z}{\theta}\right)_{\lambda',\frac{1}{\theta}}.
$$
Here $|\lambda|$ stands for the number of boxes in the diagram
$\lambda$.
\end{prop}
\begin{proof}
These relations follow immediately from definitions of
$H(\lambda,\theta)$ and $(z)_{\lambda,\theta}$.
\end{proof}
\begin{prop}\label{PropositionMSymmetries}
We have
$$
\MC_{z,z',\theta}^{(n)}(\lambda)=\MC_{-z/\theta,-z'/\theta,1/\theta}^{(n)}(\lambda').
$$
\end{prop}
\begin{proof}
Use definition of $\MC_{z,z',\theta}^{(n)}(\lambda)$, equation
(\ref{EquationVer4zmeasuren}), and apply Proposition
\ref{PropositionHH}.
\end{proof}
\begin{prop}\label{Prop1.3}
We have
$$
\sum\limits_{\lambda\in\Y_n}\MC_{z,z',\theta}^{(n)}(\lambda)=1.
$$
\end{prop}
\begin{proof}
See Kerov \cite{kerov}, Borodin and Olshanski
\cite{BO,BOHARMONICFUNCTIONS}.
\end{proof}
\begin{prop}
Expression (\ref{EquationVer4zmeasuren}) for
$\MC_{z,z',\theta}^{(n)}(\lambda)$ is strictly positive for all
$n=1,2,\ldots $ and all $\lambda\in\Y_n$ if and only if either
$z\in\C\setminus(\Z_{\leq 0}+\Zp\theta)$ and $z'=\bar z$, or, under
the additional assumption that $\theta$ is  rational, both $z, z'$
are real numbers lying in one of the intervals between two
consecutive numbers from the lattice $\Z+\Z\theta$.
\end{prop}
\begin{proof}
See Borodin and Olshanski \cite{BO}.
\end{proof}
Clearly, if the conditions in the Proposition above  are satisfied,
then $\MC_{z,z',\theta}^{(n)}$ is a probability measure defined on
$\Y_n$, as  follows from Proposition \ref{Prop1.3}.

It is convenient to mix all measures $\MC_{z,z',\theta}^{(n)}$, and
to define a new measure $\MC_{z,z',\theta,\xi}$ on
$\Y=\Y_0\cup\Y_1\cup\ldots$. Namely, let $\xi\in(0,1)$ be an
additional parameter, and set
\begin{equation}\label{EquationMzztheta}
\MC_{z,z',\theta,\xi}(\lambda)=(1-\xi)^t\xi^{|\lambda|}
\frac{(z)_{\lambda,\theta}(z')_{\lambda,\theta}}{H(\lambda,\theta)H'(\lambda,\theta)}.
\end{equation}
We also note the relation
\begin{equation}\label{EquationRelationBetweenMeasures}
\MC_{z,z',\theta,\xi}(\lambda)=(1-\xi)^t\xi^n\frac{(t)_n}{n!}\MC_{z,z',\theta,\xi}^{(n)}(\lambda),\;\;|\lambda|=n.
\end{equation}
\begin{prop} We have
$$
\sum\limits_{\lambda\in\Y}\MC_{z,z',\theta,\xi}(\lambda)=1.
$$
\end{prop}
\begin{proof}
Follows immediately from Proposition \ref{Prop1.3}.
\end{proof}
If conditions on $z,z'$ formulated in Propositions 1.2, 1.3 in
Borodin and Olshanski \cite{BO} are satisfied, then
$\MC_{z,z',\theta,\xi}(\lambda)$ is a probability measure on $\Y$.
We will refer to $\MC_{z,z',\theta,\xi}(\lambda)$ as to the
$z$-measure with the deformation (Jack) parameter $\theta$.

When both $z,z'$ go to infinity, expression
(\ref{EquationVer4zmeasuren}) has a limit
\begin{equation}\label{EquationPlancherelInfy}
\MC_{\infty,\infty,\theta}^{(n)}(\lambda)=\frac{n!\theta^{n}}{H(\lambda,\theta)H'(\lambda,\theta)}
\end{equation}
called the Plancherel measure on $\Y_n$ with general $\theta>0$.
Instead of (\ref{EquationPlancherelInfy}), sometimes it is more
convenient to consider the Poissonized Plancherel measure with
general $\theta>0$,
\begin{equation}\label{EquationPlancherelInfyMixed}
\MC_{\infty,\infty,\theta,\eta}(\lambda)=e^{-\eta^2}\left(\eta^2\right)^{|\lambda|}
\frac{\theta^{|\lambda|}}{H(\lambda,\theta)H'(\lambda,\theta)},
\end{equation}
where $\eta$ is a real parameter.
\section{Main result}
Set $ \DC_{2}(\lambda)=\left\{\lambda_i-2i\right\}.$ Thus
$\DC_{2}(\lambda)$ is an infinite subset of $\Z$ corresponding to
the Young diagram $\lambda$ containing infinitely many negative
integers. Let $X=(x_1,\ldots,x_n)$ be a subset of $\Z$ consisting
of $n$ pairwise distinct points, and define
$$
\varrho_{z,z',2,\xi}(X)=\MC_{z,z',2,\xi}\left(\{\lambda|X\subset\DC_2(\lambda)\}\right).
$$
If $\MC_{z,z',\theta,\xi}$ is positive definite, then it is a
probability measure defined on $\Y$, and
$\varrho_{z,z',\theta,\xi}(X)$ is the probability that the random
point configuration $\DC_{\theta}(\lambda)$ contains the fixed
$n$-point configuration $X=(x_1,\ldots,x_n)$. Our goal here is to
prove the following
\begin{thm}\label{Theoremzmeasure2}
a) We have
$$
\varrho_{z,z-1,2,\xi}(X)=\Pf\left[K_{z,z-1,2,\xi}(x_i,x_j)\right]_{i,j=1}^n,
$$
where the $2\times 2$ matrix valued   kernel,
$K_{z,z-1,2,\xi}(x,y)$, can be written as
$$
K_{z,z-1,2,\xi}(x,y)=\left[\begin{array}{cc}
  S_{z,z-1,2,\xi}(x+1,y+1) & -S_{z,z-1,2,\xi}(x+1,y) \\
  -S_{z,z-1,2,\xi}(x,y+1) & S_{z,z-1,2,\xi}(x,y) \\
\end{array}\right].
$$
The function $S_{z,z-1,2,\xi}(x,y)$ has the following contour
integral representation
\begin{equation}
\begin{split}
S_{z,z-1,2,\xi}(x,y)=\frac{1}{(2\pi
i)^2}\oint\limits_{\{w_1\}}&\oint\limits_{\{w_2\}}
\frac{(1+\sqrt{\xi}w_1)^{-z}(1+\sqrt{\xi}w_2)^{-z}
(1+\frac{\sqrt{\xi}}{w_1})^{z}(1+\frac{\sqrt{\xi}}{w_2})^{z}}{(w_2w_1-1)}\\
&\times\frac{(w_2-w_1)}{(w_2^2-1)(w_1^2-1)}\frac{dw_1dw_2}{w_1^{x}w_2^{y}},
\end{split}
\nonumber
\end{equation}
where $\{w_1\}$, $\{w_2\}$ are arbitrary simple contours
satisfying the conditions
\begin{itemize}
    \item both contours go around $0$ in positive direction;
    \item the unit circle is contained in the interior of the contour $\{w_1\}$, and in
    the interior of the contour $\{w_2\}$, so $|w_1|>1$,
    $|w_2|>1$.
    \item the point $\xi^{-1/2}$ lies outside both contours
    $\{w_1\}$ and $\{w_2\}$.
\end{itemize}
b) The function $S_{z,z-1,2,\xi}(x,y)$ also can be written as
\begin{equation}
\begin{split}
&S_{z,z-1,2,\xi}(x,y)=(1-\xi)^{2z}\sum\limits_{k,m=1}^{\infty}(-1)^{k+m+x+y}\xi^{\frac{k+m+x+y}{2}}\left(\Upsilon\right)_{k,m}
(z)_{k+x}(z)_{m+y}\\
&\frac{F(-z+1,-z;k+x+1;\frac{\xi}{\xi-1})}{\Gamma(k+x+1)}\frac{F(-z+1,-z;m+y+1;\frac{\xi}{\xi-1})}{\Gamma(m+y+1)},
\nonumber
\end{split}
\end{equation}
where $F(a,b;c;w)$ denotes the Gauss hypergeometric function, and
$\left(\Upsilon\right)_{k,m}$ are the matrix elements of the
matrix $\Upsilon$ defined by
\begin{equation}\label{EquationUpsilon}
\Upsilon=\left[\begin{array}{ccccccccc}
                 0 & -1 & 0 & -1 & 0 & -1 & 0 & -1 &\cdots  \\
                 1 & 0 & 0 & 0 & 0 & 0 & 0 & 0 & \cdots  \\
                 0 & 0 & 0 & -1 & 0 & -1 & 0 & -1 &  \cdots \\
                 1 & 0 & 1 & 0 & 0 & 0 & 0 & 0 & \cdots  \\
                 0 & 0 & 0 & 0 & 0 & -1 & 0 & -1 &  \cdots \\
                 1 & 0 & 1 & 0 & 1 & 0 & 0 & 0 & \cdots  \\
                 0 & 0 & 0 & 0 & 0 & 0 & 0 & -1 & \cdots  \\
                 1 & 0 & 1 & 0 & 1 & 0 & 1 & 0 & \cdots  \\
                 \vdots & \vdots & \vdots & \vdots & \vdots & \vdots & \vdots &
\vdots & \ddots
               \end{array}
\right].
\end{equation}
That is, $\Upsilon$ is an antisymmetric matrix whose entries are
defined by the relations
$$
\Upsilon(2i+1,2j+1)=\Upsilon(2i,2j)=0,\;\hbox{for any $i,j\geq
0$},
$$
$$
\Upsilon(2i+1,2j+2)=0, \;\hbox{for $0\leq i\leq j$},
$$
$$
\Upsilon(2i,2j+1)=-1, \;\hbox{for $0\leq i\leq j$}.
$$
\end{thm}
\begin{rem}
1) Since $\xi\in (0,1)$, we have $\xi/(\xi-1)<0$, and the function
$w\rightarrow F(a,b;c;w)$  is well defined on the negative
semi-axis $w<0$. The function $F(a,b;c;w)$ is not defined at
$c=0,-1,-2,\ldots ,$ but the ratio $F(a,b;c;w)/\Gamma(c)$ is well
defined for all $c\in\C$.\\
2) The matrix $\Upsilon$ appears in the study of the random matrix
ensembles of $\beta=1,4$ symmetry classes, and of their discrete
analogues, see Borodin and Strahov \cite{borodinstrahov}, Section
6, and also Borodin and Strahov \cite{borodinstrahov1}, the proof
of Lemma 6.1.
\end{rem}
\section{Pfaffian Schur measures}\label{SectionPfaffianSchur}
In the case when $z'=z-1$ the measure $\MC_{z,z-1,2,\xi}$ can be
written as a determinant, see Proposition \ref{Proposition6.1}. This
will enable us to obtain Theorem \ref{Theoremzmeasure2} as a
corollary of a more general algebraic fact.

Let $\Lambda$  denote the algebra of symmetric functions. The
algebra $\Lambda$ can be considered as the algebra of polynomials
$\C[p_1,p_2,\ldots]$ in power sums $p_1,p_2,\ldots $. Then it can
be realized, in different ways, as an algebra of functions,
depending on a specialization of the generators $p_k$. The
elements $h_k$ and $e_k$ (the complete homogeneous symmetric
functions and the elementary symmetric functions) can be
introduced through the generating series:
$$
1+\sum\limits_{k=1}^{\infty}h_ku^k=\exp\left(\sum\limits_{k=1}^{\infty}p_k\frac{u^k}{k}\right)
=\left(1+\sum\limits_{k=1}^{\infty}e_k(-u)^k\right)^{-1}.
$$
We define the generating series for $\{h_k\}$ and $\{e_k\}$ as
formal series in a complex variable $u$ by
$$
H(u)=1+\sum\limits_{k=1}^{\infty}h_ku^k,\;\;
E(u)=1+\sum\limits_{k=1}^{\infty}e_ku^k.
$$
The Schur function $s_{\lambda}$ indexed by a Young diagram
$\lambda$ can be introduced through the Jacobi-Trudi formula:
$$
s_{\lambda}=\det[h_{\lambda_i-i+j}],
$$
where, by convention, $h_0=1$, $h_{-1}=h_{-2}=\ldots=0$, and the
order of the determinant is any number greater or equal to
$l(\lambda)$.

By a specialization $\pi$ of the algebra of symmetric function
$\Lambda$ we mean a homomorphism to $\C$.

Let $\pi$  be an arbitrary specialization  of the algebra
$\Lambda$ of symmetric functions. Introduce the complex measure
\begin{equation}\label{DefinitionPfaffianSchurMeasure}
\MC_{\Schur,2}(\lambda)=\frac{1}{\ZC_{\Schur,2}}\det\left[\pi\{e_{\lambda_j-2j+i+1}\},\pi\{e_{\lambda_j-2j+i}\}\right],
\end{equation}
where $1\leq j\leq l(\lambda)$, $1\leq i\leq 2l(\lambda)$, and
$$
\ZC_{\Schur,2}=\sum\limits_{\lambda\in\Y}\det\left[\pi\{e_{\lambda_j-2j+i+1}\},\pi\{e_{\lambda_j-2j+i}\}\right]
$$ is assumed to be absolutely convergent.
Now we prove the following
\begin{thm}\label{MainTheorem}
a) For any  measure $\mathcal{M}_{\Schur, 2}$ on the set $\Y$ of
all Young diagrams, and  for any fixed subset
$X=\{x_1,\ldots,x_n\}$ of $\Z$ containing pairwise distinct points
we have the following formal series identity
\begin{equation}\label{EquationMainTheorem1}
\begin{split}
\sum\limits_{X\subset\DC_2(\lambda)}\det &\left[\pi\{e_{\lambda_j-2j+i+1}\},\pi\{e_{\lambda_j-2j+i}\}\right]\\
&=\ZC_{\Schur,2}\Pf\left[K_{\Schur,2}(x_i,x_j)\right]_{i,j=1}^n.
\end{split}
\end{equation}
Here the kernel $K_{\Schur,2}(x,y)$ can be written as
\begin{equation}\label{EquationMainTheorem2}
K_{\Schur,2}(x,y)=\left[\begin{array}{cc}
  S_{\Schur,2}(x+1,y+1) & -S_{\Schur,2}(x+1,y) \\
  -S_{\Schur,2}(x,y+1) & S_{\Schur,2}(x,y) \\
\end{array}\right],
\end{equation}
and the function $S_{\Schur,2}(x,y)$ admits the following integral
representation
\begin{equation}
\begin{split}
&S_{\Schur,2}(x,y)\\
&=\frac{1}{(2\pi i)^2}\oint\limits_{\{w_1\}}\oint\limits_{\{w_2\}}
\frac{dw_1dw_2}{w_1^{x}w_2^{y}}\frac{\pi\left\{E(w_1)\right\}}{\pi\left\{E(w_1^{-1})\right\}}
\frac{\pi\left\{E(w_2)\right\}}{\pi\left\{E(w_2^{-1})\right\}}\frac{(w_2-w_1)}{(w_2w_1-1)(w_2^2-1)(w_1^2-1)},
\end{split}
\nonumber
\end{equation}
where
$\pi\left\{E(w)\right\}=1+\sum\limits_{k=1}^{\infty}\pi\{e_k\}w^k$,
and where $\{w_1\}$, $\{w_2\}$ are simple contours which both go
around $0$ in positive direction, they  are chosen in such a way
that they do not include  the possible poles at
$\pi\left\{E(w_1^{-1})\right\}, \pi\left\{E(w_2^{-1})\right\}$, and
$|w_1|>1$, $|w_2|>1$.\\
b) The function $S_{\Schur,2}(x,y)$ also admits the representation
\begin{equation}\label{EquationSSchur2asSUM}
S_{\Schur,2}(x,y)=\sum\limits_{k,m=1}^{\infty}\Phi_k(x)\left(\Upsilon\right)_{k,m}\Phi_m(y),
\end{equation}
where the functions $\Phi_k$ are given by
\begin{equation}\label{Equation3.9}
\Phi_k(x)=\frac{1}{2\pi
i}\oint\limits_{\{w\}}\frac{\pi\{E(w)\}}{\pi\{E(\frac{1}{w})\}}\frac{dw}{w^{k+x+1}}.
\end{equation}
In the formulae just written above $\{w\}$ is a simple contour
which goes around $0$ in positive direction, and
$\left(\Upsilon\right)_{k,m}$ are the matrix elements of the
matrix $\Upsilon$ defined by equation (\ref{EquationUpsilon}).
\end{thm}
\begin{rem}
For a specific specialization $\pi$ of $\Lambda$ one can
investigate conditions under which equations
(\ref{EquationMainTheorem1}), (\ref{EquationMainTheorem2}) become
numerical equalities.
\end{rem}
Define the correlation function of $\MC_{\Schur,2}$ by
$$
\varrho_{\Schur,2}(X)=\MC_{\Schur,2}\left(\{\lambda|X\subset\DC_2(\lambda)\}\right)
$$
If $\MC_{\Schur,2}$ is positive definite, then it is a probability
measure defined on $\Y$, and $\varrho_{\Schur,2}(X)$ is the
probability that the random point configuration $\DC_{2}(\lambda)$
contains the fixed $n$-point configuration $X=(x_1,\ldots,x_n)$.
By Theorem \ref{MainTheorem} we have
$$
\varrho_{\Schur,2}(X)=\Pf\left[K_{\Schur,2}(x_i,x_j)\right]_{i,j=1}^n,
$$
where the correlation kernel, $K_{\Schur,2}(x,y)$, is given by
formulae (\ref{EquationMainTheorem1}),
(\ref{EquationMainTheorem2}).
\section{Proof of Theorem \ref{MainTheorem}}
Let $X=(x_1,\ldots, x_n)$ be a fixed subset of $\Z$ consisting of
$n$ pairwise  distinct points. Let us pick an integer $N$ such
that $N>n$, and such that $X$ is a subset of
$\{-2N,-2N+1,\ldots\}$. Define
$$
\varrho_{\Schur,2}^{(N)}(X)=\frac{\MC_{\Schur,2}\left\{\lambda|\lambda\in\Y(N),
X\subset\{\lambda_i-2i\}\right\}}{\MC_{\Schur,2}\left\{\lambda|\lambda\in\Y(N)\right\}}.
$$
Here $\Y(N)$ denotes the subset of $\Y$ consisting of Young
diagrams whith number of rows  less or equal to $N$. Clearly,
$\varrho_{\Schur,2}^{(N)}(X)$ converges to $\varrho_{\Schur,2}(X)$
in the statement of Theorem \ref{MainTheorem} as
$N\longrightarrow\infty$.

For a given Young diagram $\lambda$ from $\Y(N)$ we define
$l_j=\lambda_j-2j$, $1\leq j\leq N$ (if the length $l(\lambda)$ of
the diagram $\lambda$
 is less then $N$ we set $\lambda_{l(\lambda)+1}=0,
\lambda_{l(\lambda)+2}=0,\ldots, \lambda_{N}=0$). We also
introduce functions $\phi_i(l)$ and $\psi_i(l)$ on $\Z$ (where
$1\leq i\leq 2l(\lambda)$) by formulae
$\phi_i(l)=\pi\{e_{l+i+1}\}$, $\psi_i(l)=\pi\{e_{l+i}\}$, where
 $\pi$ is a specialization of the algebra $\Lambda$ of symmetric
 functions.
With this notation we can rewrite $\varrho_{\Schur,2}^{(N)}(X)$ as
follows
\begin{equation}\label{EquationVarrho(M)}
\varrho_{\Schur,2}^{(N)}(X)=\frac{\underset{X\subset\{l_1,\ldots,l_N\}}{\sum\limits_{l_1>\ldots
>l_{N}\geq -2N}}\det\left[\phi_i(l_j),\psi_i(l_j)\right]}{\sum\limits_{l_1>\ldots
>l_{N}\geq
-2N}\det\left[\phi_i(l_j),\psi_i(l_j)\right]},
\end{equation}
where $1\leq i\leq 2N,\; 1\leq j\leq N$. By methods of Random
Matrix Theory (see, for example, Tracy and Widom \cite{tracy},
Section 8) we can represent $\varrho_{\Schur,2}^{(N)}(X)$ as a
Pfaffian of a $2\times 2$ matrix valued kernel. Namely, we can
obtain the following Lemma
\begin{lem}\label{LemmaTracyWidomFormula}
Suppose that $X$ and $N$ are chosen as described above, and
suppose that $\varrho^{(N)}(X)$ is defined by equation
(\ref{EquationVarrho(M)}), where $\phi_i(l)$, $\psi_i(l)$  are now
arbitrary functions of finite support defined on $\Z$, and $1\leq
i\leq N$. The following identity holds true
$$
\varrho^{(N)}(X)=\Pf[K^{(N)}(x_i,x_j)]_{i,j=1}^n,
$$
where $K^{(N)}(x,y)$ is a $2\times 2$ matrix valued kernel defined
by the formula
\begin{equation}
K^{(N)}(x,y)=\left[\begin{array}{cc}
  \sum\limits_{i,j=1}^{2N}\phi_i(x)((M(N))^{-1})_{i,j}\phi_j(y) & -\sum\limits_{i,j=1}^{2N}\phi_i(x)((M(N))^{-1})_{i,j}\psi_j(y) \\
  -\sum\limits_{i,j=1}^{2N}\psi_i(x)((M(N))^{-1})_{i,j}\phi_j(y) & \sum\limits_{i,j=1}^{2M}\psi_i(x)((M(N))^{-1})_{i,j}\psi_j(y) \\
\end{array}\right],
\nonumber
\end{equation}
where
$(M(N))_{i,j}=\sum\limits_{x=-{2N}}^{+\infty}\left(\phi_i(x)\psi_j(x)-\phi_j(x)\psi_i(x)\right)$,
$1\leq i,j\leq 2N$.
\end{lem}
In our case, Lemma \ref{LemmaTracyWidomFormula} (with
$\phi_i(x)=\pi(e_{x+i+1})$ and $\psi_i(x)=\pi(e_{x+i})$, where
$x\in\Z$) gives the following representation for the correlation
function $\varrho_{\Schur,2}^{(N)}(X)$
\begin{equation}
\varrho_{\Schur,2}^{(N)}(X)=\Pf\left[K_{\Schur,2}^{(N)}(x_i,x_j)\right]_{i,j=1}^n,
\end{equation}
where
$$
K_{\Schur,2}^{(N)}(x,y)=\left[\begin{array}{cc}
  S_{\Schur,2}^{(N)}(x+1,y+1) & -S_{\Schur,2}^{(N)}(x+1,y) \\
  -S_{\Schur,2}^{(N)}(x,y+1) & S_{\Schur,2}^{(N)}(x,y) \\
\end{array}
  \right],
$$
and where the function $S_{\Schur,2}^{(N)}(x,y)$ is defined by
$$
S_{\Schur,2}^{(N)}(x,y)=\sum\limits_{i,j=1}^{2N}\pi(e_{x+i+1})((M(N))^{-1})_{i,j}\pi(e_{y+j+1}).
$$
In the formulae  above the matrix $M(N)$ is a $2N\times 2N$ matrix
whose matrix coefficients are defined by the formula
$$
(M(N))_{i,j}=\sum\limits_{x=-2N}^{+\infty}\left(\pi(e_{x+i+1})\pi(e_{x+j})-\pi(e_{x+j+1})\pi(e_{x+i})\right),\;
1\leq i,j\leq 2N.
$$
As $\pi\{e_{-k}\}=0$ for $k=1,2,\ldots$ we can rewrite the formula
for the matrix coefficients $(M(N))_{i,j}$ as follows
$$
(M(N))_{i,j}=\sum\limits_{x=-\infty}^{+\infty}\left(\pi(e_{x+i+1})\pi(e_{x+j})-
\pi(e_{x+i})\pi(e_{x+j+1})\right), \;1\leq i,j\leq 2N.
$$
Clearly, the matrix $M(N)$ can be understood as the $N$th
principal minor of the infinite matrix $M^{\infty}$ whose matrix
elements are defined by the formula
\begin{equation}\label{EquationMInfinity}
(M^{\infty})_{i,j}=\sum\limits_{x=-\infty}^{+\infty}\left(\pi(e_{x+i+1})\pi(e_{x+j})-
\pi(e_{x+i})\pi(e_{x+j+1})\right), \;1\leq i,j.
\end{equation}
Rains \cite{rains} has shown that for matrices of this type
$$
\underset{N\rightarrow\infty}{\lim
}\left((M(N))^{-1}-(M^{\infty})^{-1}\right)_{i,j}=0
$$
for any fixed $i,j$, see Lemma 2.1, equations (3.12), (4.23) in
Rains \cite{rains}. We then deduce that the kernel
$K_{\Schur,2}(x,y)$ can be written as
$$
K_{\Schur,2}(x,y)=\left[\begin{array}{cc}
  S_{\Schur,2}(x,y) & -S_{\Schur,2}(x+1,y) \\
  -S_{\Schur,2}(x,y) & S_{\Schur,2}(x,y) \\
\end{array}
  \right],
$$
where
\begin{equation}\label{EquationSchur2Sum}
S_{\Schur,2}(x,y)=\sum\limits_{i,j=1}^{\infty}\pi(e_{x+i})((M^{\infty})^{-1})_{i,j}\pi(e_{y+j}),
\end{equation}
and the semi-infinite matrix $M^{\infty}$ is defined by equation
(\ref{EquationMInfinity}).

It is convenient to rewrite formula (\ref{EquationMInfinity}). Set
$k=x+i+j$. Note that $k$ can be any (negative or positive)
integer. We have $x+i=k-j$, and $x+j=k-i$. Then
$(M^{\infty})_{i,j}$ can be rewritten as
$$
(M^{\infty})_{i,j}=\sum\limits_{k=-\infty}^{+\infty}\left(\pi(e_{k-j+1})\pi(e_{k-i})-
\pi(e_{k-j})\pi(e_{k-i+1})\right), \;1\leq i,j.
$$
As soon as $\pi\{e_{-k}\}=0$, for $k=1,2,\ldots$, we can also
represent $(M^{\infty})_{i,j}$  as
$$
(M^{\infty})_{i,j}=\sum\limits_{k=1}^{+\infty}\left(\pi(e_{k-j+1})\pi(e_{k-i})-
\pi(e_{k-j})\pi(e_{k-i+1})\right), \;1\leq i,j.
$$

Let $\varphi(z)$ be a formal series of the form
$$
\varphi(z)=\sum\limits_{n\in\Z}\varphi_nz^n.
$$
It is convenient to introduce the following notation. Let
$[z^n]\varphi$ stand for the coefficient of $z^n$ in $\varphi$,
and let $T\varphi$ denote the Toeplitz semi-infinite matrix
defined by $\varphi$. The matrix elements of $T\varphi$ are
$$
(T\varphi)_{i,j}=\varphi_{j-i}=[z^{j-i}]\varphi=\frac{1}{2\pi
i}\oint\limits_{\{w\}}w^{j-i}\varphi(w)\frac{dw}{w},
$$
where $\{w\}$ is an arbitrary simple contour which goes around $0$
in positive direction.

With the above notation rewrite the expression for the matrix
elements $(M^{\infty})_{i,j}$ as follows
$$
(M^{\infty})_{i,j}=\sum\limits_{k=1}^{+\infty}\left(T\pi
\{E(z)\}\right)_{j,k+1}\left(T\pi\{E(z)\}\right)_{i,k}-
\left(T\pi\{E(z)\}\right)_{i,k+1}\left(T\pi \{E(z)\}\right)_{j,k},
$$
where $1\leq i,j$.

Up to now we have used essentially the same arguments as in Rains
\cite{rains}. In what follows we show that $S_{\Schur,2}(x,y)$
(equation (\ref{EquationSchur2Sum})) admits the contour integral
representation as in Theorem \ref{MainTheorem}.

Introduce a semi-infinite matrix $D$ defining its matrix elements
by
$$
(D)_{k,m}=\delta_{k+1,m}-\delta_{k,m+1},\;\; 1\leq k,m.
$$
Then we obtain
$$
(M^{\infty})_{i,j}=\sum\limits_{k,m=1}^{+\infty}\left(T\pi
\{E(z)\}\right)_{i,k}(D)_{k,m}\left(T\pi\{E(z)\}\right)_{j,m},\;
1\leq i,j.
$$
Therefore, we have $M^{\infty}=T\pi \{E(z)\}D\left(T\pi
\{E(z)\}\right)^t$. Note that $D$ has the following form
$$
D=\left[\begin{array}{ccccccccc}
  0 & 1 & 0 & 0 & 0 & 0 & 0 & 0 & \cdots \\
  -1 & 0 & 1 & 0 & 0 & 0 & 0 & 0 & \cdots \\
  0 & -1 & 0 & 1 & 0 & 0 & 0 & 0 & \cdots \\
  0 & 0 & -1 & 0 & 1 & 0 & 0 & 0 & \cdots \\
  0 & 0 & 0 & -1 & 0 & 1 & 0 & 0 & \cdots \\
0 & 0 & 0 & 0 & -1 & 0 & 1 & 0 & \cdots \\
 0 & 0 & 0 & 0 & 0 & -1 & 0 & 1 & \cdots \\
0 & 0 & 0 & 0 & 0 & 0 & -1 & 0 & \cdots \\
\vdots & \vdots & \vdots & \vdots & \vdots & \ddots & \ddots &\ddots &\ddots \\
\end{array}\right].
$$
This matrix has inverse, which is denoted by $\Upsilon$, whose
explicit form is given by equation (\ref{EquationUpsilon}).

Now we can write $(M^{\infty})^{-1}=\left[\left(T\pi
\{E(z)\}\right)^t\right]^{-1}\Upsilon\left[T\pi
\{E(z)\}\right]^{-1}$, and we obtain the following expression for
the matrix elements of $(M^{\infty})^{-1}$
\begin{equation}\label{Equation3.7}
\begin{split}
&\left((M^{\infty})^{-1}\right)_{i,j}\\
&=\sum\limits_{k,m=1}^{\infty}\left(\left[\left(T\pi
\{E(z)\}\right)^t\right]^{-1}\right)_{i,k}\left(\Upsilon\right)_{k,m}\left(\left[T\pi
\{E(z)\}\right]^{-1}\right)_{m,j},
\end{split}
\end{equation}
where $i,j\geq 1$.
\begin{lem} For any integers $j,m$ such that $j,m\geq 1$
\begin{equation}\label{Equation3.8}
\left(\left[T\pi
\{E(z)\}\right]^{-1}\right)_{j,m}=[z^{m-j}]\left(\frac{1}{\pi\{E(z)\}}\right).
\end{equation}
\end{lem}
\begin{proof}
We need to check that
\begin{equation}\label{Syma}
\sum\limits_{k=1}^{\infty}\left(\left[T\pi
\{E(z)\}\right]^{-1}\right)_{j,k}\left(\left[T\pi
\{E(z)\}\right]\right)_{k,m}=\delta_{j,m}
\end{equation}
remains to be valid, if we replace $\left(\left[T\pi
\{E(z)\}\right]^{-1}\right)_{j,k}$ by
$[z^{k-j}]\left(\frac{1}{\pi\{E(z)\}}\right)$. Note that
$$
[z^k]\left(\frac{1}{\pi\{E(z)\}}\right)=0,\;\mbox{for}\; k\leq -1,
$$
and
$$
[z^k]\left(\pi\{E(z)\}\right)=0,\;\mbox{for}\; k\leq -1.
$$
Now the left-hand side of  equation (\ref{Syma}) can be written as
\begin{equation}
\begin{split}
&\sum\limits_{k=1}^{\infty}[z^{k-j}]\left(\frac{1}{\pi\{E(z)\}}\right)\left(\left[T\pi
\{E(z)\}\right]\right)_{k,m}=\sum\limits_{k=1}^{\infty}[z^{k-j}]\left(\frac{1}{\pi\{E(z)\}}\right)[z^{m-k}]\left(\pi\{E(z)\}\right)\\
&=\sum\limits_{k=j}^{\infty}[z^{k-j}]\left(\frac{1}{\pi\{E(z)\}}\right)[z^{m-k}]\left(\pi\{E(z)\}\right)
=\sum\limits_{n=0}^{\infty}[z^{n}]\left(\frac{1}{\pi\{E(z)\}}\right)[z^{m-n-j}]\left(\pi\{E(z)\}\right).
\end{split}
\nonumber
\end{equation}
Taking into account the formula
$$
[z^k]\left(A(z)B(z)\right)=\sum\limits_{n=0}^{\infty}[z^n]A(z)[z^{k-n}]B(z),
$$
we obtain that
$\sum\limits_{k=1}^{\infty}[z^{k-j}]\left(\frac{1}{\pi\{E(z)\}}\right)\left(\left[T\pi
\{E(z)\}\right]\right)_{k,m}$ equals $\delta_{j,m}$.
\end{proof}
Now we are ready to derive the contour integral representation for
the kernel $S_{\Schur,2}(x,y)$.  Equations (\ref{Equation3.7}),
(\ref{Equation3.8}) imply formula (\ref{EquationSSchur2asSUM}),
with $\Phi_k$  defined by the formula
\begin{equation}\label{EquationPhiSum}
\Phi_k(x)=\sum\limits_{i=1}^{\infty}[z^{x+i}]\pi\{E(z)\}[z^{i-k}]\left(\frac{1}{\pi\{E(z)\}}\right),
\; k\geq 1.
\end{equation}
It is possible to represent $\Phi_k(x)$ as a contour integral, see
equation (\ref{Equation3.9}). To see this rewrite formula
(\ref{EquationPhiSum}) as follows
\begin{equation}
\Phi_k(x)=
\sum\limits_{i=-\infty}^{\infty}[z^{x+i}]\pi\{E(z)\}[z^{i-k}]\left(\frac{1}{\pi\{E(z)\}}\right),
\nonumber
\end{equation}
where we have used the fact that the formal  series
$\left(\pi\{E(z)\}\right)^{-1}$ does not contain terms of the form
$a_nz^{-n}, n=1,2,\ldots$. Changing the index of the summation,
and applying formulae
$$
[z^k]A(z^{-1})=[z^{-k}]A(z),
$$
and
$$
[z^k]\left(A(z)B(z)\right)=\sum\limits_{i=-\infty}^{\infty}[z^i]A(z)[z^{k-i}]B(z),
$$
we find
$$
\Phi_k(x)=[z^{x+k}]\left[\frac{\pi\{E(z)\}}{\pi\{E(z^{-1})\}}\right].
$$
This formula is equivalent to equation (\ref{Equation3.9}). Using
(\ref{Equation3.9}) we can represent $S_{\Schur,2}(x,y)$ as
follows
$$
S_{\Schur,2}(x,y)=\frac{1}{(2\pi
i)^2}\oint\limits_{\{w_1\}}\oint\limits_{\{w_2\}}\frac{\pi\{E(w_1)\}}{\pi\{E(\frac{1}{w_1})\}}
\frac{\pi\{E(w_2)\}}{\pi\{E(\frac{1}{w_2})\}}\left(\sum\limits_{k,m=1}^{\infty}\frac{(\Upsilon)_{k,m}}{w_1^kw_2^m}\right)
\frac{dw_1}{w_1^{x+1}}\frac{dw_2}{w_2^{y+1}},
$$
where $\{w_1\}$, $\{w_2\}$ are simple contours which go around $0$
in positive direction.

The proof of Theorem \ref{MainTheorem} is accomplished by the
following
\begin{prop} Assume that $|w_1|>1$, $|w_2|>1$, and let
$\Upsilon$ be a semi-infinite matrix defined by equation
(\ref{EquationUpsilon}). Then
$$
\sum\limits_{k,m=1}^{\infty}\frac{(\Upsilon)_{k,m}}{w_1^kw_2^m}
=\frac{w_2w_1(w_2-w_1)}{(w_2w_1-1)(w_2^2-1)(w_1^2-1)}.
$$
\end{prop}
\begin{proof}
It is convenient to represent the sum in the left-hand side of the
equation just written above as
$$
\left[\begin{array}{cccc}
  \frac{1}{w_1} & \frac{1}{w_1^2} & \frac{1}{w_1^3} & \ldots \\
\end{array}\right]\Upsilon\left[\begin{array}{c}
  \frac{1}{w_2} \\
  \frac{1}{w_2^2} \\
  \frac{1}{w_2^3} \\
  \vdots \\
\end{array}\right].
$$
Using the explicit form of $\Upsilon$ (see equation
(\ref{EquationUpsilon})) we find that this matrix product equals
\begin{equation}
\begin{split}
&\frac{1}{w_1}\left(-\frac{1}{w_2^2}-\frac{1}{w_2^4}-\frac{1}{w_2^6}-\frac{1}{w_2^8}-\ldots\right)\\
&+\frac{1}{w_1^2}\frac{1}{w_2}\\
&+\frac{1}{w_1^3}\left(-\frac{1}{w_2^4}-\frac{1}{w_2^6}-\frac{1}{w_2^8}-\frac{1}{w_2^{10}}-\ldots\right)\\
&+\frac{1}{w_1^4}\frac{1}{w_2}+\frac{1}{w_1^4}\frac{1}{w_2^3}\\
&+\frac{1}{w_1^5}\left(-\frac{1}{w_2^6}-\frac{1}{w_2^8}-\frac{1}{w_2^{10}}-\frac{1}{w_2^{12}}-\ldots\right)\\
&+\frac{1}{w_1^6}\frac{1}{w_2}+\frac{1}{w_1^6}\frac{1}{w_2^3}+\frac{1}{w_1^6}\frac{1}{w_2^5}\\
&+\frac{1}{w_1^7}\left(-\frac{1}{w_2^8}-\frac{1}{w_2^{10}}-\frac{1}{w_2^{12}}-\frac{1}{w_2^{14}}-\ldots\right)\\
&+\frac{1}{w_1^8}\frac{1}{w_2}+\frac{1}{w_1^8}\frac{1}{w_2^3}+\frac{1}{w_1^8}\frac{1}{w_2^5}+\frac{1}{w_1^8}\frac{1}{w_2^7}\\
&+\ldots\\
&=-\frac{1}{w_1w_2^2}\frac{1}{1-\frac{1}{w_2^2}}-\frac{1}{w_1^3w_2^4}\frac{1}{1-\frac{1}{w_2^2}}
-\frac{1}{w_1^5w_2^6}\frac{1}{1-\frac{1}{w_2^2}}-\ldots\\
&+\frac{1}{w_1^2w_2}\frac{1}{1-\frac{1}{w_1^2}}+\frac{1}{w_1^4w_2^3}\frac{1}{1-\frac{1}{w_1^2}}
+\frac{1}{w_1^6w_2^5}\frac{1}{1-\frac{1}{w_1^2}}+\ldots\\
&=\frac{w_1w_2(w_2-w_1)}{(w_1w_2-1)(w_1^2-1)(w_2^2-1)}.
\end{split}
\nonumber
\end{equation}
\end{proof}\section{Proof of Theorem \ref{Theoremzmeasure2}}
\begin{prop}\label{Proposition6.1}
We have
\begin{equation}\label{Equationz59}
\begin{split}
\MC_{z,z-1,2,\xi}(\lambda)
=&(1-\xi)^{\frac{z(z-1)}{2}}\xi^{|\lambda|}\left(z\right)_{\lambda,2}
\left(z-1\right)_{\lambda,2}\\
&\times\det\left[\frac{1}{(\lambda_j-2j+i+1)!},\frac{1}{(\lambda_j-2j+i)!}\right],
\end{split}
\end{equation}
where $1\leq i\leq 2l(\lambda)$, and $1\leq j\leq l(\lambda)$.
\end{prop}
\begin{proof}
Using the explicit formula for $H(\lambda,2)H'(\lambda,2)$ given in
the proof of Lemma 3.5 in Borodin and Olshanski \cite{BO}, we
rewrite $\left(H(\lambda, 2)H'(\lambda, 2)\right)^{-1}$ in a
determinantal form
\begin{equation}\label{Equationz2}
\begin{split}
\frac{1}{H(\lambda, 2)H'(\lambda, 2)}&=
\det\left[\frac{1}{(\lambda_j-2j+i+1)!},\frac{1}{(\lambda_j-2j+i)!}\right],
\end{split}
\end{equation}
where $1\leq i\leq 2l(\lambda)$, $1\leq j\leq l(\lambda)$. Then the
formula in the statement of the Proposition follows from the
definition of  $\MC_{z,z',\theta,\xi}(\lambda)$ (equation
(\ref{EquationMzztheta})), and from equation (\ref{Equationz2}).
\end{proof}
\begin{prop}\label{Propositionz52}
Let $\pi_z$ be the specialization of the algebra $\Lambda$ of
symmetric functions defined by
\begin{equation}\label{zrealization}
\pi_{z}(e_k)=\xi^{\frac{k}{2}}\frac{(z-k+1)(z-k+2)\ldots z}{k!}
\end{equation}
Then the following formula holds true
$$
\MC_{z,z-1,2,\xi}(\lambda)=(1-\xi)^{\frac{z(z-1)}{2}}
\det\left[\pi_{-z}\{e_{\lambda_j-2j+i+1}\},\pi_{-z}\{e_{\lambda_j-2j+i}\}\right],
$$
where $1\leq i\leq 2l(\lambda)$, and $1\leq j\leq l(\lambda)$.
\end{prop}
\begin{proof}
Assume first that $z=N$. Then, using equation
(\ref{zrealization}), we find  that
\begin{equation}\label{Equationz510}
\pi_{z=N}\{e_k\}=\xi^{\frac{k}{2}}\frac{(N-k+1)(N-k+2)\ldots
N}{k!}=\xi^{\frac{k}{2}}\left(\begin{array}{c}
  N \\
  k \\
\end{array}
\right).
\end{equation}
Let us compute the determinant
$$
\det\left[\pi_{z}\{e_{\lambda_j-2j+i+1}\},\pi_{z}\{e_{\lambda_j-2j+i}\}\right],
$$
where $1\leq i\leq 2l(\lambda)$,  $1\leq j\leq l(\lambda)$, and
the parameter $z$ equals $N$. Taking into account equation
(\ref{Equationz510}) we find
\begin{equation}
\begin{split}
&\det\left[\pi_{z=N}\{e_{\lambda_j-2j+i+1}\},\pi_{z=N}\{e_{\lambda_j-2j+i}\}\right]\\
&=\det\left[\xi^{\frac{\lambda_j-2j+i+1}{2}}\left(\begin{array}{c}
  N \\
  \lambda_j-2j+i+1 \\
\end{array}\right), \xi^{\frac{\lambda_j-2j+i}{2}}\left(\begin{array}{c}
  N \\
  \lambda_j-2j+i \\
\end{array}\right)\right]\\
&=\xi^{|\lambda|}\frac{N!(N+1)!}{(N-\lambda_1)!(N-\lambda_1+1)}
\frac{(N+2)!(N+3)!}{(N-\lambda_2+2)!(N-\lambda_2+3)}\times\ldots\\
&\times\frac{(N+2l(\lambda)-2)!(N+2l(\lambda)-1)!}{(N-\lambda_{l(\lambda)}+2l(\lambda)-2)!(N-\lambda_{l(\lambda)}+2l(\lambda)-1)!}
\det\left[\frac{1}{(\lambda_j-2j+i+1)!},\frac{1}{(\lambda_j-2j+i)!}\right].
\end{split}
\nonumber
\end{equation}
This equation can be further rewritten as follows
\begin{equation}
\begin{split}
&\det\left[\pi_{z=N}\{e_{\lambda_j-2j+i+1}\},\pi_{z=N}\{e_{\lambda_j-2j+i}\}\right]\\
&=\xi^{|\lambda|}\prod\limits_{i=1}^{l(\lambda)}(-N-2i+2)_{\lambda_i}(-N-1-2i+2)_{\lambda_i}
\det\left[\frac{1}{(\lambda_j-2j+i+1)!},\frac{1}{(\lambda_j-2j+i)!}\right],
\end{split}
\nonumber
\end{equation}
where $1\leq i\leq 2l(\lambda)$, and $1\leq j\leq l(\lambda)$. As
soon as
$(z)_{\lambda,2}=\prod\limits_{i=1}^{l(\lambda)}(z-2i+2)_{\lambda_i}$,
we see that the equation
\begin{equation}\label{Equationz11}
\begin{split}
&\det\left[\pi_{z}\{e_{\lambda_j-2j+i+1}\},\pi_{z}\{e_{\lambda_j-2j+i}\}\right]\\
&=\xi^{|\lambda|}(-z)_{\lambda}(-z-1)_{\lambda}
\det\left[\frac{1}{(\lambda_j-2j+i+1)!},\frac{1}{(\lambda_j-2j+i)!}\right]
\end{split}
\end{equation}
holds true when $z=N$. By an analytic continuation, equation
(\ref{Equationz11}) also holds true for an arbitrary complex $z$.
Comparing equations (\ref{Equationz59}) and (\ref{Equationz11}) we
obtain the formula in the statement of the Proposition.
\end{proof}
Now we are ready to complete the proof of Theorem
\ref{Theoremzmeasure2}. Proposition \ref{Propositionz52} implies
that $\MC_{\Schur,2}$ is reduced to $\MC_{z,z-1,2,\xi}$ when the
specialization $\pi$ in the definition of $\MC_{\Schur,2}$ is
$\pi_{-z}$. Therefore, the correlation function of
$\MC_{z,z-1,2,\xi}$ can be obtained from Theorem
\ref{MainTheorem}. We only need to compute $\pi_{-z}\{E(w)\}$, and
$\pi_{-z}\{E(\frac{1}{w})\}$. When $\pi=\pi_z$, where $\pi_{z}$ is
the specialization of the algebra $\Lambda$ of symmetric functions
defined by equation (\ref{zrealization}), we find
\begin{equation}\label{EquationpiE(z)}
\pi_{z}\{E(w)\}=\sum\limits_{n=0}^{\infty}(-1)^n\frac{(-z)_n}{n!}(\sqrt{\xi}w)^n=(1+\sqrt{\xi}w)^{z},
\end{equation}
provided that the condition $|w|<\xi^{-1/2}$ is satisfied. We also
obtain
\begin{equation}\label{EquationpiE(z)1}
\pi_{z}\{E(\frac{1}{w})\}=\sum\limits_{n=0}^{\infty}(-1)^n\frac{(-z)_n}{n!}(\frac{\sqrt{\xi}}{w})^n=
\left(1+\frac{\sqrt{\xi}}{w}\right)^{z},
\end{equation}
where it is assumed that $|w|>\xi^{1/2}$. If
$\xi^{1/2}<|w|<\xi^{-1/2}$ both equations (\ref{EquationpiE(z)})
and (\ref{EquationpiE(z)1}) hold true. As  $0<\xi<1$, we can
choose the contours $\{w_1\}$, $\{w_2\}$ as in Theorem
\ref{MainTheorem}, with the additional requirement that
$\xi^{-1/2}$ lies outside both contours $\{w_1\}$ and $\{w_2\}$.
This gives formulae in Theorem \ref{Theoremzmeasure2}, a).

With the choice of the contour $\{w\}$ as in the statement of
Theorem \ref{Theoremzmeasure2}, the formula for the function
$\Phi_k(x)$ (see equation (\ref{Equation3.9})) takes the form
$$
\Phi_k(x)=\frac{1}{2\pi
i}\oint\limits_{\{w\}}(1+\sqrt{\xi}w)^{-z}(1+\frac{\sqrt{\xi}}{w})^{z}\frac{dw}{w^{k+x+1}}.
$$
The integral in the right-hand side of the formula above can be
computed, and the function $\Phi_k(x)$ can be expressed as
$$
\Phi_k(x)=(-1)^{k+x}\xi^{\frac{k+x}{2}}(1-\xi)^{z}(z)_{k+x}\frac{F(-z+1,-z;k+x+1;\frac{\xi}{\xi-1})}{\Gamma(k+x+1)},
$$
where $F(a,b;c;u)$ denotes the Gauss hypergeometric function.
Equation (\ref{EquationSSchur2asSUM}) gives the representation for
the function $S_{z,z-1,2,\xi}(x,y)$ as in Theorem
\ref{Theoremzmeasure2}, b).\qed
\section{The Pfaffian process defined by the the Plancherel measure with the Jack parameter $\theta=2$}
\begin{thm}\label{TheoremPlancherelMeasure2}
a) For $\MC_{\infty,\infty,2,\eta}$ defined by equation
(\ref{EquationPlancherelInfyMixed}) the point configurations
$\DC_2(\lambda)=\{\lambda_i-2i\}_{i=1}^{\infty}\subset\Z$ form a
Pfaffian point process. This means that the correlation function
$\varrho_{\infty,\infty,2,\eta}$ of $\MC_{\infty,\infty,2,\eta}$
defined by
$$
\varrho_{\infty,\infty,2,\eta}=\MC_{\infty,\infty,2,\eta}\left(\{\lambda|X\subset\DC_2(\lambda)\}\right),
$$
(where $X=(x_1,\ldots,x_n)$ is any fixed subset of $\Z$ consisting
of $n$ distinct points) can be represented as a Pfaffian of a
$2\times 2$ matrix valued kernel. Namely,
$$
\varrho_{\infty,\infty,2,\eta}=\Pf\left[K_{\infty,\infty,2,\eta}(x_i,x_j)\right]_{i,j=1}^n,
$$
where
$$
K_{\infty,\infty,2,\eta}(x,y)=\left[\begin{array}{cc}
  S_{\infty,\infty,2,\eta}(x+1,y+1) & -S_{\infty,\infty,2,\eta}(x+1,y) \\
  -S_{\infty,\infty,2,\eta}(x,y+1) & S_{\infty,\infty,2,\eta}(x,y) \\
\end{array}\right],
$$
and
\begin{equation}
\begin{split}
&S_{\infty,\infty,2,\eta}(x,y)\\
&=\frac{1}{(2\pi i)^2}\oint\limits_{\{w_1\}}\oint\limits_{\{w_2\}}
\frac{dw_1dw_2}{w_1^{x}w_2^{y}}\exp[\sqrt{2}\eta(w_1-w_1^{-1}+w_2-w_2^{-1})]\frac{(w_2-w_1)}{(w_1w_2-1)(w_1^2-1)(w_2^2-1)}.
\end{split}
\nonumber
\end{equation}
Here $\{w_1\}$, $\{w_2\}$ are simple contours which both go around
$0$ in positive direction, and they are chosen in such a
way  that $|w_1|>1$, $|w_2|>1$.\\
b) The function $S_{\infty,\infty,2,\eta}(x,y)$ can be written as
$$
S_{\infty,\infty,2,\eta}(x,y)=\sum\limits_{k,m=1}^{\infty}J_{k+x}(2\sqrt{2}\eta)\left(\Upsilon\right)_{k,m}
J_{m+x}(2\sqrt{2}\eta),
$$
where $\Upsilon$ is defined by equation (\ref{EquationUpsilon}),
and where $J_{\nu}(x)$ is the Bessel function, see Ref.
\cite{functions}, 7.2.4.
\end{thm}
\begin{proof}
Consider the specialization $\pi_{\infty}^{(\sqrt{2}\eta)}$ of the
algebra $\Lambda$ of symmetric functions defined by
\begin{equation}\label{EquationPlancherelSpecialization}
\pi_{\infty}^{(\sqrt{2}\eta)}(p_k)=\left\{%
\begin{array}{ll}
    \eta, & k=1, \\
    0, & k=2,3,\ldots . \\
\end{array}%
\right.
\end{equation}
In this specialization we have
\begin{equation}\label{EquationPiEk}
\pi_{\infty}^{(\sqrt{2}\eta)}\{e_k\}=\frac{(\sqrt{2}\eta)^k}{k!},\;\mbox{and}\;
\pi_{\infty}^{(\sqrt{2}\eta)}\{E(w)\}=\exp{(\sqrt{2}\eta w)}.
\end{equation}
Next observe that $\MC_{\Schur,2}$ is reduced to
$\MC_{\infty,\infty,2,\eta}$, if the specialization $\pi$ in the
definition of $\MC_{\Schur,2}$ coincides with the specialization
$\pi_{\infty}^{(\sqrt{2}\eta)}$. This follows from equations
(\ref{DefinitionPfaffianSchurMeasure}), and from the fact that
\begin{equation}\label{EquationPlancherelasDeterminant}
\MC_{\infty,\infty,2,\eta}(\lambda)=e^{-\eta^2}\det\left[
\frac{(\sqrt{2}\eta)^{\lambda_j-2j+i+1}}{(\lambda_j-2j+i+1)!},
\frac{(\sqrt{2}\eta)^{\lambda_j-2j+i}}{(\lambda_j-2j+i)!}\right],
\nonumber
\end{equation}
where $1\leq i\leq 2l(\lambda)$, $1\leq j\leq l(\lambda)$. We then
conclude that the formula for $\varrho_{\Plancherel,2}$ must have
the same form as that for $\varrho_{\Schur,2}$. As
$\pi=\pi_{\Plancherel}^{(\sqrt{2}\eta)}$ we check (taking into
account equation (\ref{EquationPiEk})) that the formulae  obtained
in Theorem \ref{MainTheorem} are reduced to the formulae
 in the statement of
Theorem \ref{TheoremPlancherelMeasure2}.
\end{proof}

\end{document}